\documentclass[conference]{IEEEtran}

\usepackage{setspace}
\usepackage{amsmath}
\usepackage{amssymb}
\usepackage{mathrsfs}
\usepackage{amsthm}
\usepackage{multirow}
\usepackage{color}
\usepackage{longtable}
\usepackage{array}
\usepackage{url}
\usepackage{comment}
\usepackage{enumerate}
\usepackage{eucal}

\usepackage{algorithm}
\usepackage[noend]{algorithmic}
\usepackage{cite}
\usepackage[table]{xcolor}

\usepackage{stmaryrd}

\usepackage[top=0.6in, bottom=0.6in, left=0.5in, right=0.5in]{geometry}

\theoremstyle{plain}
\newtheorem{theorem}{Theorem}
\newtheorem{corollary}[theorem]{Corollary}
\newtheorem{lemma}[theorem]{Lemma}
\newtheorem{proposition}[theorem]{Proposition}

\theoremstyle{definition}
\newtheorem{definition}[theorem]{Definition}
\newtheorem{example}[theorem]{Example}

\newtheorem{construction}{Construction}

\newcommand{\F}{\mathbb{F}}

\newcommand{\A}{\mathtt{A}}
\newcommand{\T}{\mathtt{T}}
\newcommand{\C}{\mathtt{C}}
\newcommand{\G}{\mathtt{G}}
\newcommand{\GC}{$\mathtt{GC}$}

\DeclareMathAlphabet{\mathbfsl}{OT1}{ppl}{b}{it} %{OT1}{cmr}{bx}{it}

\newcommand{\va}{\mathbfsl{a}}
\newcommand{\vb}{\mathbfsl{b}}
\newcommand{\vm}{\mathbfsl{m}}

\newcommand{\vp}{\mathbfsl{p}}
\newcommand{\vu}{\mathbfsl{u}}
\newcommand{\vv}{\mathbfsl{v}}
\newcommand{\vw}{\mathbfsl{w}}

\newcommand{\vc}{\mathbfsl{c}}

\newcommand{\cA}{\mathcal{A}}
\newcommand{\cB}{\mathcal{B}}
\newcommand{\cC}{\mathcal{C}}

\newcommand{\bsigma}{{\pmb{\sigma}}}

\newcommand{\bphi}{{\pmb{\phi}}}

\newcommand{\wt}{{\rm wt}}
\newcommand{\shift}{\bsigma}

\newcommand{\etal}{{\em et al.}}

\newcommand{\floor}[1]{{\left\lfloor #1\right\rfloor}}
\newcommand{\ceil}[1]{{\left\lceil #1\right\rceil}}

\newcommand{\pspan}[1]{{\left\langle #1\right\rangle}}
\newcommand{\bbracket}[1]{{\left\llbracket #1\right\rrbracket}}
\newcommand{\equivcyclic}{\underset{\text{cyc}}{\sim}}

\title{Efficient and Explicit Balanced Primer Codes\\[-4mm]}

\author{
   \IEEEauthorblockN{
	Yeow Meng Chee,
	Han Mao Kiah, and
	Hengjia Wei}
   \IEEEauthorblockA{
   School of Physical and Mathematical Sciences,
	Nanyang Technological University, Singapore\\
email: \{{ymchee}, {hmkiah}, {hjwei}\}@ntu.edu.sg\\[-2mm]
 } 
}

\begin{document}
\date{}

\maketitle

\hspace*{-10pt}\begin{abstract}

To equip DNA-based data storage with random-access capabilities,
Yazdi \etal{} (2018) prepended DNA strands with specially chosen address sequences called primers and provided certain design criteria for these primers.
We provide explicit constructions of error-correcting codes that are suitable as primer addresses and
equip these constructions with efficient encoding algorithms.

Specifically, our constructions take cyclic or linear codes as inputs and produce sets of primers with similar error-correcting capabilities. 
Using certain classes of BCH codes, we obtain infinite families of primer sets of length $n$, minimum distance $d$ with $(d+1) \log_4 n +O(1)$ redundant symbols.
Our techniques involve reversible cyclic codes (1964), an encoding method of Tavares \etal{} (1971) and Knuth's balancing technique (1986).
In our investigation, we also construct efficient and explicit binary balanced error-correcting codes and codes for DNA computing.
\end{abstract}

\section{Introduction}

Advances in synthesis and sequencing technologies have made DNA macromolecules an attractive medium for digital information storage. 
Besides being biochemically robust, DNA strands offer ultrahigh storage densities of $10^{15}-10^{20}$ bytes per gram of DNA, as demonstrated in recent experiments (see \cite[Table 1]{Yazdi2017.portable}).
Therefore, in recent years, new error models were proposed and novel coding schemes were constructed by various authors (see \cite{Yazdi.etal:2015b} for a survey).

In this paper, we study the problem of {\em primer design}.
To introduce random-access and rewriting capabilities into DNA-based data storage, 
Yazdi \etal{} developed an architecture that allows selective access to encoded DNA strands through the process of  {\em hybridization}.
Their technique involves prepending information-carrying DNA strands with specially chosen address sequences called primers.
Yazdi \etal{} provided certain design considerations for these primers \cite{Yazdietal2018.mu}
and also, verified the feasibility of their architecture in a series of experiments \cite{Yazdi.etal:2015, Yazdi.etal:2015b}.

We continue this investigation and 
provide efficient and explicit constructions of error-correcting codes that are suitable as primer addresses.
Our techniques include novel modifications of {\em Knuth's balancing technique} \cite{Knuth1986} and 
involve the use of {\em reversible cyclic codes} \cite{Massey1964}. 
We also revisit the work of Tavares \etal{} \cite{Tavaresetal1971} that efficient encodes messages into cyclic classes of a cyclic code and adapt their method for our codes. 
We note that reversible cyclic codes have been studied in another coding application for DNA computing.
It turns out our techniques can be also modified to improve code constructions in the latter application.
%However, due to space constraints, our exposition is focussed on the code constructions for primers and 
%we refer interested readers to the full version for other codes for DNA-related applications \cite{arxiv}.

\section{Preliminary and Contributions}

Let $\F_q$ denote the finite field of size $q$. 
%A {\it $q$-ary code} $\cC$ of length $n$ is a collection of words from $\F_q^n$.
Two cases of special interest are $q = 2$ and $q = 4$. 
In the latter case, we let $\omega$ denote a primitive element of $\F_4$ and
identify the elements of  $\F_4$ with the four DNA bases $\Sigma=\{\A,\C,\T,\G\}$.
Specifically, 
\[0\leftrightarrow \A,\quad 1\leftrightarrow \T, \quad\omega\to \C, \quad\omega+1\leftrightarrow \G.\]
Hence, for an element $x\in\F_4$, its Watson-Crick complement corresponds to $x+1$.  

Let $n$ be a positive integer. 
Let $[n]$ denote the set $\{1,2,\ldots, n\}$, 
while $\bbracket{n}$ denotes the set $\{0,1,\ldots,n-1\}$.
For a word $\va = (a_1, \ldots, a_n) \in \F_q^n$, 
let  $\va[i]$ denote the $i$th symbol $a_i$ and  $\va[i,j]$ denote the subword of $\va$ starting at position $i$ and ending at position $j$. In other words,
\begin{equation*}
\begin{split}
\va[i,j] = & \begin{cases} (a_i,a_{i+1},\ldots,a_j), \textup{\ \ if $i\leq j$;} \\
(a_j,a_{j-1},\ldots,a_i), \textup{\ \ if $i>j$.}\\
\end{cases}
\end{split}
\end{equation*}
Moreover, the {\it reverse} of $\va$, denoted as $\va^r$, is $(a_n,a_{n-1}, \ldots, a_1)$; 
the {\it complement  $\overline{\va}$} of $\va$ is $(\overline{a}_1,\overline{a}_2, \ldots, \overline{a}_n)$,
where $\overline{x}=x+1$ for $x\in \F_2$ or $x\in\F_4$;
%(for $q=2$, $\overline{0}=1$ and $\overline{1}=0$; for  $q=4$, $\overline{\A}=\T$, $\overline{\T}=\A$, $\overline{\C}=\G$, and $\overline{\G}=\C$); 
and the {\it reverse-complement $\va^{rc}$} of $\va$ is $\overline{\va^r}$.
%$(\overline{a}_n,\overline{a}_{n-1},\ldots, \overline{a}_1)$. 

For two words $\va$ and  $\vb$, we use $\va\vb$ to denote the concatenation of  $\va$ and $\vb$,
and $\va^\ell$ to denote the sequence of length $\ell n$ comprising $\ell$ copies of $\va$.

A {\it $q$-ary code} $\cC$ of length $n$ is a collection of words from $\F_q^n$.
For two words $\va$ and $\vb$ of the same length, 
we use $d(\va,\vb)$ to denote the Hamming distance between them. 
A code $\cC$ has  {\it minimum Hamming distance $d$} 
if any two distinct codewords in $\cC$ is at least distance $d$ apart.
Such a code is denoted as an {$(n,d)_q$-code}. 
Its {\em size} is given by $|\cC|$, 
while its {\em redundancy} is given by $n-\log_q |\cC|$.
An {\em $[n,k,d]_q$-linear code} is an $(n,d)_q$-code  that 
is also an $k$-dimension vector subspace of $\F_q^n$.
Hence, an $[n,k,d]_q$-linear code has redundancy $n-k$.

\subsection{Cyclic and Reversible Codes}

For a vector $\va\in \F_q^n$, let $\shift^i(\va)$ be the vector obtained by cyclically shifting the components of $\va$ to right $i$ times.
%In other words, if $\va=(a_1,a_2,\ldots,a_n)$, then 
So, $\shift^1(\va)=(a_n,a_1,a_2,\ldots,a_{n-1})$.  An {\it $[n, k, d]_q$-cyclic code} $\cC$ is 
 an $[n,k,d]_q$-linear code that is closed under cyclic shifts. In other words, $\va\in \cC$ implies $\shift^1(\va) \in \cC$.  
% A cyclic code can be partitioned into its {\it cyclic classes}, 
% where a cyclic class comprises codewords which are cyclic shifts of one another. 

Cyclic codes are well-studied because of their rich algebraic structure.
In the theory of cyclic codes (see for example, MacWilliams and Sloane \cite[Chapter 7]{macwilliams1977theory}), 
we identify a word $\vc=(c_i)_{i\in\bbracket{n}}$ of length $n$
with the polynomial $\sum_{i=0}^{n-1} c_iX^i$.   
Given a cyclic code $\cC$ of length $n$ and dimension $k$,
there exists a unique monic polynomial $g(X)$ of degree $n-k$ such that 
$\cC$ is given by the set $\{m(X)g(X) : \deg m < k\}$.
The polynomial $g(X)$ is referred to as the {\em generator polynomial} of $\cC$ and 
we write $\cC=\pspan{g(X)}$.
We continue this discussion on this algebraic structure in Section~\ref{sec:encoding},
where we exploit certain polynomial properties for efficient encoding.

When $d$ is fixed, there exists a class of Bose-Chaudhuri-Hocquenghem (BCH) codes 
that are cyclic codes whose redundancy is asymptotically optimal.

%\begin{theorem}[Primitive narrow-sense BCH codes {\cite[Theorem 10]{Aly:2007}}]\label{primitiveBCH}
%Fix $m\ge 1$ and $d\le q^{\floor{m/2}}$. Set $n=q^m-1$ and $t=\ceil{(d-1)(1-/q)}$.
%There exists an $[n,k,d]_q$-cyclic code $\cC$ with $k= n-tm$.
%In other words, $\cC$ has redundancy at most $t\log_q n$.
%\end{theorem}

\begin{theorem}[Primitive narrow-sense BCH codes {\cite[Theorem 10]{Aly:2007}}]\label{primitiveBCH}
Fix $m\ge 1$ and $2\le d\le 2^m-1$. Set $n=2^m-1$ and $t=\ceil{(d-1)/2}$.
There exists an $[n,k,d]_2$-cyclic code $\cC$ with $k\ge  n-tm$.
In other words, $\cC$ has redundancy at most $t\log_2 (n+1)$.
\end{theorem}

%Hence, for all values of $k$, if we shorten an appropriate BCH code, 
%we obtain a $[k+p,k,d]_2$-linear code with $p\le t\log_2 k+1$, where $t=\floor{(d-1)/2}$.
%{\color{red} NEED TO CHECK}

A cyclic code $\cC$ is called {\it reversible} if $\va \in \cC$ implies $\va^r \in \cC$. 
A reversible cyclic  code is also known as an LCD cyclic code and has been studied extensively  \cite{CLietal2017,Lietal2017,Massey1964,TzengHartmann}.
In this paper,  reversible cyclic codes containing the all-one vector $1^n$ are of particular interest. Suppose that $\cC$ is one such  code. 
Then for any codeword $\va\in\cC$,  both its complement $\overline{\va}=\va+1^n$ and 
its reverse-complement $\va^{rc}=\va^r+1^n$ belong to $\cC$. 

Recently, Li \etal \cite{Lietal2017} explored two other classes of BCH codes and 
determined their minimum distances and dimensions. 
These codes are reversible cyclic and contain the all-one vector. 
%Here we list their result for the code $\cC_{(q,n,2\delta-1,\frac{n+1}{2}-\delta+1)}$ with $q$ being even. 

\begin{theorem}[Li \etal \cite{Lietal2017}]\label{thm-LCDwithone}
Let $m\ge 2$, $m\not=3$ and $1\le \tau\le \ceil{m/2}$.
Let $q$ be even and set $n=q^m-1$ and $d=q^\tau-1$.
There exists an $[n,k,d]_q$-reversible cyclic code that contains $1^n$
and has dimension
\begin{equation*}
k = 
\begin{cases} 
n-(d-q+1)m, &\mbox{if $m\geq 5$ is odd and $\tau=\frac{m+1}2$}; \\
n-(d-1)m,  &\mbox{otherwise}.\\
\end{cases}
\end{equation*}
In other words, $\cC$ has redundancy at most $(d-1)\log_q (n+1)$.
\end{theorem}

\subsection{Balanced Codes}

A binary word of length $n$ is  {\it balanced} if $\floor{n/2}$ or $\ceil{n/2}$ bits are zero, while 
a quaternary word of length $n$ is {\it \GC-balanced} if $\floor{n/2}$ or $\ceil{n/2}$ symbols are either $\G$ or $\C$. 
A binary (or quaternary) code is {\it balanced} (resp. {\em \GC-balanced}) 
if all its codewords are balanced (resp. \GC-balanced).

Motivated by applications in laser disks, Knuth \cite{Knuth1986} studied balanced binary codes and  
proposed an efficient method to encode an arbitrary binary message to a binary balanced codeword
by introducing  $\log_2 n$ redundant bits. 
Recently, Weber \etal \cite{Weberetal2012} extended Knuth's scheme to include error-correcting capabilities. 
Specifically, their construction takes two input codes of distance $d$: 
a linear code of length $n$ and a short {\em balanced} code $\cC_p$;
and outputs a long balanced code of distance $d$.
Even though the balanced code $\cC_p$ is only required to be size $n$,
it is unclear how to find one efficiently, especially when $d$ grows with $n$.

On the other hand, \GC-balanced codes have been extensively studied in the context of DNA computing and DNA-based storage (see \cite{limbachiyaetal2016, Yazdi.etal:2015b, ImminkCai:2018} for a survey).
However, most constructions are based on search heuristics or apply to a restricted set of parameters.
Recently, Yazdi \etal \cite{Yazdietal2018.mu} introduced the coupling construction (Lemma~\ref{coupling}) that 
takes two binary error-correcting codes, one of which is balanced, as inputs 
and outputs a \GC-balanced error-correcting code.
As with the construction of Weber \etal{} \cite{Weberetal2012}, it is unclear how to find the balanced binary error-correcting code efficiently.

In this work, we {\em avoid these requirements of additional balanced codes}. 
Specifically, we provide construction that takes a binary cyclic code (or two binary linear codes)
and outputs a binary balanced code (resp. a \GC-balanced code) with error-correcting capabilities. 

%When $d$ is constant, their method constructs a balanced $(n,d)_2$-codes with redundancy 
%$(d/2+1) \log_2 n +(d+1)/{2}\log_2\log_2 n$. 

\subsection{Primer Codes}

In order to introduce random access to DNA-based data storage systems, 
Yazdi \etal \cite{Yazdietal2018.mu} proposed the following criteria for the design of primer addresses.
%Specifically, the following properties were introduced.

\begin{definition}
A code $\cC$ of length $n$ is {\it $\kappa$-weakly mutually uncorrelated ($\kappa$-WMU)} if 
for all $\ell \geq \kappa$, no proper prefix of length $\ell$ of a codeword appears as 
a suffix of another codeword (including itself). 
In other words, for any two codewords $\va,\vb \in \cC$, not necessarily distinct, and $\kappa\leq \ell\leq n$,
\[ \va[1,\ell]\ne \vb[n-\ell+1,n].\]
When $\cC$ is $1$-WMU, we say that $\cC$ is mutually uncorrelated (MU).
\end{definition}

\begin{definition}
A code $\cC$ of length $n$ is said to 
{\it avoid primer dimer byproducts of effective length $f$} ($f$-APD) if 
the reverse complement and the complement of any substring of length $f$ in a codeword does not appear in as a substring of another codeword (including itself).
In other words, for any two codewords $\va,\vb \in \cC$, not necessarily distinct, and $1\leq i,j\leq n+1-f$, we have 
\[\overline{\va}[i,{i+f-1}] \notin \{\vb[j,{j+f-1}], \vb[{j+f-1},j]\}.\]
\end{definition}

For primer design in DNA-based storage, 
WMU codes are desired to be \GC-balanced, have large Hamming distance and avoid primer dimer byproducts. 

\begin{definition}
A code $\cC\in \F_q^n$ is an {\it $(n,d;\kappa,f)_q$-primer code} if the following are satisfied:
\begin{enumerate}
\item[(P1)] $\cC$ is an $(n,d)_q$-code;
\item[(P2)] $\cC$ is $\kappa$-WMU;
\item[(P3)] $\cC$ is an $f$-APD code.
\end{enumerate}
Furthermore, if $\cC$ is  balanced or \GC-balanced, then $\cC$ is an {\it $(n,d;\kappa,f)_q$-balanced primer code}. 
\end{definition}

Yazdi \etal \cite{Yazdietal2018.mu} provided a number of constructions for WMU codes which satisfy some combinations of the constraints (P1), (P2) and (P3).  
In particular, Yazdi \etal{} provided the following coupling construction.

\begin{lemma}[Coupling Construction - Yazdi \etal{} \cite{Yazdietal2018.mu}]\label{coupling}
For $i\in [2]$, let $\cC_i$ be an $(n,d_i)_2$-code of size $M_i$.
Define the map $\Psi: \F_2^n\times \F_2^n \to \Sigma^n$ such that
$\Psi(\va,\vb)=\vc$ where for $i\in [n]$,
\begin{equation*}
c_i=  \begin{cases}
\A, &\text{if $a_ib_i=00$;} \\
\T, &\text{if $a_ib_i=01$;} \\
%\C, \textup{\ \ if $(a_i,b_i)=(1,0)$;} \\
%\G, \textup{\ \ if $(a_i,b_i)=(1,1)$.}
\end{cases}
\quad
c_i=  \begin{cases}
%\A, \textup{\ \ if $(a_i,b_i)=(0,0)$;} \\
%\T, \textup{\ \ if $(a_i,b_i)=(0,1)$;} \\
\C, &\text{if $a_ib_i=10$;} \\
\G, &\text{if $a_ib_i=11$.}
\end{cases}
\end{equation*}
%\begin{equation*}
%c_i=  \begin{cases}\A, \textup{\ \ if $(a_i,b_i)=(0,0)$;} \\
%\T, \textup{\ \ if $(a_i,b_i)=(0,1)$;} \\
%\C, \textup{\ \ if $(a_i,b_i)=(1,0)$;} \\
%\G, \textup{\ \ if $(a_i,b_i)=(1,1)$.}
%\end{cases}
%\end{equation*}
Then the code $\cC\triangleq \{\Psi(\va,\vb): \va\in \cC_1, \vb\in \cC_2\}$ is an $(n,d)_4$-code of size $M_1M_2$, where $d=\min\{d_1,d_2\}$.
Furthermore,
\begin{enumerate}[(i)]
\item if $\cC_1$ is balanced, $\cC$ is \GC-balanced;
\item if $\cC_2$ is $\kappa$-WMU, then $\cC$ is also $\kappa$-WMU;
\item if $\cC_2$ is an $f$-APD code, then $\cC$ is also an $f$-APD code.
\end{enumerate}
\end{lemma}

%Among other constructions, 
Yazdi \etal{} also provided an iterative construction for primer codes 
satisfying all the constraints, {\em i.e.} balanced primer codes.
However, the construction requires a short balanced primer code
and a collection of subcodes, some of which disjoint.
Hence, it is unclear whether the code can be constructed efficiently and whether efficient encoding is possible.

In this work, we provide constructions that take cyclic, reversible cyclic or linear codes as inputs and 
produce primer or balanced primer codes as outputs.
Using known families of cyclic codes given by Theorems~\ref{primitiveBCH}~and~\ref{thm-LCDwithone},
we obtain infinite families of primer codes and provide explicit upper bounds on the redundancy.
We also describe methods that efficiently encode into these codewords.

\subsection{Our Contributions}
In this paper, we study balanced codes, primer codes and other related coding problems. 
Our contributions are as follow:
{
%\color{red}
\begin{enumerate}[A.]
\item In Section~\ref{sec:balanced}, we propose efficient methods to construct both balanced and \GC-balanced error-correcting codes. Unlike previous methods that require short balanced error-correcting codes, our method uses only cyclic and linear codes as inputs.
Furthermore, our method always increases the redundancy only by $\log_2 n+1$ (where $n$ is the block length), regardless of the value of the minimum distance.
\item In Section~\ref{sec:primer}, we provide three constructions of primer codes. 
For general parameters, the first construction produces a class of $(n,d;\kappa, f)_4$-balanced primer codes whose redundancy is $(d+1)\log_4 n+O(1)$, while the other two rely on cyclic codes and use less redundancy albeit for a specific set of parameters.
In particular, we have a class of $(n,d;\kappa, \kappa)_4$-balanced primer codes with 
redundancy $(d+1)\log_4 (n+1)$.
\item In Section~\ref{sec:DNAcomputing}, we construct codes for DNA computing.
In particular, we provide a class of \GC-balanced $(n,d)_4$-DNA computing codes with redundancy 
$(d+1)\log_4(n+1)$.
\item In Section~\ref{sec:encoding}, we adapt the technique of Tavares \etal{} to efficiently encode messages into codes constructed in this paper.
\end{enumerate}
}
%
%\begin{enumerate}[A.]
%\item In Section~\ref{Sec-balancedcode}, we propose an efficient  method to construct balanced codes with error-correcting ability. 
%  $(n,d)$ cyclic code to produce a balanced code of  minimum distance $\geq d$, while  increasing the redundancy only by $\log_2 n+1$.
%\item In Section~\ref{Sec-balancedprimercode}, we modify the previous construction of WMU codes. Then using the modified construction and  the balanced error-correcting codes,   we construct a class of $(n,d;\kappa, f)_4$ balanced primer codes. When $n$ is large  and $d$ is fixed, the redundancy of  such codes is asymptotically equal to $(d+1)\log_4 n+O(1)$.
%\item We propose a new construction of primer codes in  Section~\ref{Sec-primercode}.  The construction works well even for small $n$ and can produce  a class of $(n,d;\kappa, f)_4$ primer codes of redundancy at most $(d+1)\log_4 (n +1)+1$. 
%\item In Section~\ref{Sec-DNAcode}, we modify our construction of balanced codes to obtain DNA codes. We construct  an infinite  class of DNA codes which satisfy all the four constraints (C1) to (C4) and have $\G\C$-content ${n}/{2}$ and redundancy $d\log_4 (n+1)+1.5$.  
%\end{enumerate}

\section{Balanced Error-Correcting Codes}
\label{sec:balanced}

The celebrated Knuth's balancing technique \cite{Knuth1986} is a linear-time algorithm 
that maps a binary message of length $m$ to a balanced word of length approximately $m+\log m$.
The technique first finds an index $z$ such that flipping the first $z$ bits yields a balanced word $\vc$.
Then Knuth appends a short balanced word $\vp$ that represents the index $z$.
Hence, $\vc\vp$ is the resulting codeword and 
the redundancy of the code is equal to the length of $\vp$ which is approximately $\log m$.
The crucial observation demonstrated by Knuth is that such an index $z$ always exists and $z$ is commonly referred to as the balancing index.
%\subsection{Background}

Recently, Weber \etal \cite{Weberetal2012}  modified Knuth's balancing technique to endow the code with error-correcting capabilities. 
Their method requires two error-correcting codes as inputs:
an $(m,d)_2$ code $\cC_m$ and a short $(p,d)_2$ balanced code $\cC_p$ where $|\cC_p|\ge m$.
Given a message, they first encode it into a codeword $\vm\in\cC_m$.
Then they find the balancing index $z$ of $\vm$ and flip the first $z$ bits to obtain a balanced $\vc$.
Using $\cC_p$, they encode $z$ into a balanced word $\vp$ and the resulting codeword is $\vc\vp$.
Since both $\cC_m$ and $\cC_p$ has distance $d$, the resulting code has minimum distance $d$.

Now, this method introduces $p$ {\em additional} redundant bits and
since $p$ is necessarily at least $d$, the method introduces more than $\log_2 n$ bits of redundancy when $d$ is big. 
Furthermore, the method requires the existence of a short balanced code $\cC_p$. 
%{\color{red}Now, in order for this encoding scheme to be efficient, 
%it remains to provide efficient encoding methods for both $\cC_m$ and $\cC_p$.
%Since $\cC_m$ is not constrained, we can simply choose $\cC_m$ to be a linear code.
%In contrast, it is not clear for the short balanced code $\cC_p$.
%When $d$ is fixed, we may choose $p=O(\log m)$ and 
%efficient encoding of $\cC_p$ can be accomplished using a lookup table.
%However, when $d$ grows with $n$, we need to find efficient means of encoding $\cC_p$.
%}
We overcome this obstacle in our next two constructions.
Specifically, Construction~\ref{constr-bin-bal} and~\ref{constr-gc-bal} 
require only a cyclic code and a linear codes, respectively.
Both constructions do {\em not} require short balanced codes and
introduces only $\log_2 n +1$ {\em additional} bits of redundancy, regardless the value of $d$.

\subsection{Binary Balanced Error-Correcting Codes}

Let $n$ be odd. In contrast with Knuth's balancing technique, 
we {\em always flip the first $(n+1)/2$ bits} of a word $\va$. 
However, this does not guarantee a balanced word. 
Nevertheless, if we consider all cyclic shifts of $\va$, i.e. $\shift^i(\va)$ for $i\in\bbracket{n}$, 
then flipping the first $(n+1)/2$ bits of one of these shifts must yield a balanced word.

Formally, let $\phi: \F_2^n \to \F_2^n$ be the map where $\phi(\va)=\va+1^{(n+1)/2}0^{(n-1/2)}$.
In other words, the map $\phi$ flips the first $(n+1)/2$ bits of $\va$.
For $\va\in \F_2^n$, denote its Hamming weight as $\wt(\va)$.
Let $\wt_1(\va)$ be the Hamming weight of the first $(n+1)/{2}$ bits
and $\wt_2(\va)$ be the Hamming weight of the last $(n-1)/{2}$ bits.
So, we have $\wt(\va)=\wt_1(\va)+\wt_2(\va)$. 
We have the following crucial lemma.

\begin{lemma}\label{lemma-flip}
Let $n$ be  odd. For $\va\in \F_2^n$, we can find  $i\in \llbracket n \rrbracket$ such that  $\phi(\shift^i(\va))$ has weight either $(n-1)/2$ or $(n+1)/2$. 
\end{lemma}

\begin{proof} 
Let $\va'=\shift^{(n+1)/2}(\va)$. Then the first $(n-1)/2$ bits of $\va'$ are exactly the last $(n-1)/2$ bits of $\va$ and so $\wt_2(\va)\leq \wt_1(\va')\leq \wt_2(\va)+1$. 

We first consider the case when $\wt(\va)$ is even. Assume that $\wt(\va)=2w$. If $\wt_1(\va)\leq w$, then $\wt_1(\va')\geq \wt_2(\va)=2w-\wt_1(\va)\geq w$. Note that shifting the components of $\va$ once  only increases or decreases the value of $\wt_1(\va)$ by at most one.  It follows that we can find an integer $i$ such that $\wt_1(\shift^i(\va))=w$, and so 
\begin{align*}
\wt(\phi(\shift^i(\va)))& = \wt_1(\phi(\shift^i(\va)))+\wt_2(\phi(\shift^i(\va)))\\
 &=   ((n+1)/2-w)+w=(n+1)/2.
\end{align*} 
Similarly, if $\wt_1(\va)>w$, since $\wt_1(\va')\leq \wt_2(\va)+1=2w-\wt_1(\va)+1\leq w$, we can still find $i$ such that $\wt_1(\shift^i(\va))=w$ and  $\wt(\phi(\shift^i(\va)))=(n+1)/2.$

Next, we assume that the weight is odd, or, $\wt(\va)=2w+1$. If $\wt_1(\va)<w+1$, then $\wt_1(\va')\geq \wt_2(\va)=2w+1-\wt_1(\va)\geq w+1$; if $\wt_1(\va)\geq w+1$, then $\wt_1(\va')\leq \wt_2(\va)+1=2w+1-\wt_1(\va)+1\leq w+1$. In both cases we can always find  $i$ such that $\wt_1(\shift^i(\va))=w+1$, and so
\begin{align*}
\wt(\phi(\shift^i(\va)))& = \wt_1(\phi(\shift^i(\va)))+\wt_2(\phi(\shift^i(\va)))\\
 &=   ((n+1)/2-w-1)+w=(n-1)/2.\qedhere
\end{align*} 
\end{proof}

\noindent{\bf Remark.}  
In Lemma~\ref{lemma-flip}, we show that we can balance some shift of $\va$ by flipping its first $(n+1)/2$ bits. 
In fact, we can also balance a shift of $\va$ (not necessary the same shift) by flipping its first $(n-1)/2$ bits. 
This observation is used in the construction of DNA computing codes.
\vspace{1mm}

Before we describe our construction, we introduce the notion of cyclic equivalence classes.
Given a cyclic code $\cB$ of length $n$, we define the following equivalence relation: $\va\equivcyclic\vb$ if and only if 
$\va=\shift^i(\vb)$ for some $i\in [n]$; and partition the codewords $\cB$ into classes.
We use $\cB/\equivcyclic$ to denote a set of representatives.

\begin{construction}\label{constr-bin-bal}
Let $n$ be an odd integer.

{\sc Input}: An $[n,k,d]_2$-cyclic code $\cB$.\\
{\sc Output}: A balanced $(n+1,d')_2$-code $\cC$ of size at least $2^k/n$ where $d'=2\ceil{d/2}$.

\begin{itemize}
\item Let $\vu_1, \vu_2, \ldots, \vu_m$ be the set of representatives $\cB/\equivcyclic$. 
\item For each $\vu_i$, find $j_i\in [n]$ such that 
$\phi(\shift^{j_i}(\vu_i))$ has weight $(n-1)/2$ or $(n+1)/2$. 
 \item For $i\in [m]$, append a check bit to $\phi(\shift^{j_i}(\vu_i))$ so that its weight is $(n+1)/2$ and denote the modified vector as $\vv_i$. In other words, 
\begin{equation*}
\begin{split}
\vv_i = & \begin{cases} \phi(\shift^{j_i}(\vu_i))0, \textup{\ \ if $\wt(\phi(\shift^{j_i}(\vu_i)))=\frac{n+1}{2}$;} \\
\phi(\shift^{j_i}(\vu_i))1, \textup{\ \ if $\wt(\phi(\shift^{j_i}(\vu_i)))=\frac{n-1}{2}$.}\\
\end{cases}
\end{split}
\end{equation*}
\item Set $\cC=\{\vv_i:1\leq i\leq m\}$.
\end{itemize}
%Let $\cC'=\{\vv_i:1\leq i\leq m\}$. 
\end{construction}

\begin{theorem}\label{thm-balancedecc}
Construction~\ref{constr-bin-bal} is correct. 
In other words,  $\cC$ is a balanced $(n+1,2\ceil{d/2})_2$-code of size at least $2^k/n$.
\end{theorem}
%\begin{theorem}\label{thm-balancedecc}
%$\cC'$ is an $(n,2 \lceil d/2\rceil)$ balanced code of  size $m$, where $m\geq  |\cC|/n$.
%\end{theorem}

\begin{proof}
It is easy to see that $\cC$ is a balanced code of length $n+1$ and size $m$.  
Since the $m$ cyclic classes are pairwise disjoint and each of them consists of at most $n$ codewords,
we have that  $m \geq |\cB|/n=2^k/n$.

Since $\cB$ is an $[n,k,d]_2$-cyclic code and the map $\phi$ does not change the distance between any two vectors, the minimum distance of $\cC$ is at least $d$. Moreover, when $d$ is odd, the minimum distance is at least $d+1$, as the distance between any two binary balanced words is even.
\end{proof}

Let $d$ be even and set $t=d/2-1$.
If we apply Construction~\ref{constr-bin-bal} to the family of primitive narrow-sense BCH $[n',k,d]_2$-cyclic codes, where $n'=2^m-1=n-1$.
we obtain a family of balanced codes with redundancy at most $(t+1)\log_2 n+1$.

\begin{corollary}\label{cor:bin-balanced}
Let $d$ be even. 
There exists a family of $(n,d)_2$-balanced codes with redundancy at most $(t+1)\log_2 n+1$, where $t=d/2-1$.
\end{corollary}

In contrast, if we apply the technique of Weber \etal \cite{Weberetal2012} to the same family of codes, 
the balanced $(n,d)_2$-codes have redundancy approximately $(t+1)\log_2 n +(t+1/2)\log_2 \log_2 n$.
Hence, we reduce the redundancy by $(t+1/2)\log_2 \log_2 n$ bits.

Finally, we consider the encoding complexity for our construction. 
Given a vector $\vu$, we can find in linear time the index $i$ such that $\wt(\phi(\shift^i(\vu)))\in \{(n-1)/2,(n+1)/2\}$.
Thus, it remains to provide an efficient method to enumerate a set of representatives for the cyclic classes. 
This problem was solved completely by Tavares \etal \cite{Tavaresetal1971,Tavaresetal1973} and 
the solution uses the polynomial representation of cyclic codewords.
Furthermore, the encoding method can be adapted for Constructions~\ref{constr-primer-3} and~\ref{constr-DNA} in the later sections.
Hence, we review Tavares' method in detail and discuss our modifications in Section~\ref{sec:encoding}.

\subsection{\GC-Balanced Error-Correcting Codes}

A direct application of the coupling construction in Lemma~\ref{coupling} and Corollary~\ref{cor:bin-balanced}
yields a family of \GC-balanced $(n,d)_4$-codes with redundancy at most $d\log_4 n$.
However, this construction requires cyclic codes of length $n-1$.

The following construction removes the need for cyclic codes.

\begin{construction} \label{constr-gc-bal}\hfill

{\sc Input}: An $[n+p,n,d]_2$-linear code $\cA$ and

\hspace{12mm}an $(n,d)_2$-code $\cB$ of size $2^pnM$.\\
{\sc Output}: A balanced $(n,d)_4$-code $\cC$ code of size $2^nM$.
\begin{itemize}
\item Given $\vm\in\F_2^n$, let $j_\vm$ be the balancing index of $\vm$ and $\va_m$ be the corresponding balanced word of length $n$.
\item Consider a systematic encoder for $\cA$. For $\va_\vm\in\F_2^{n}$, 
let $\va_\vm\vp_\vm$ be the corresponding codeword in $\cA$.
\item Finally, since $\cB$ is of size $2^pnM$, we may assume without loss of generality an encoder
$\bphi_\cB: [M]\times [n]\times \F_2^p\to \cB$. We set $\vb_\vm=\bphi(i,j_\vm,\vp_\vm)$.
\item Set $\cC\triangleq\{ \Psi(\va_\vm,\vb_\vm) :\, \vm\in \F_2^n, i\in [M] \}$.
\end{itemize}
\end{construction}

\begin{theorem}\label{thm-gc-balanced}
Construction~\ref{constr-gc-bal} is correct. 
In other words,  $\cC$ is a \GC-balanced $(n,d)_4$-code of size at least $2^nM$.
\end{theorem}

\begin{proof}The size of $\cC$ follows from its definition. 

For all words $\vc=\Psi(\va,\vb)$ in $\cC$, since $\va$ is balanced, we have that $\vc$ is \GC-balanced.
Hence, $\cC$ is \GC-balanced.

Finally, to prove that $\cC$ has distance $d$, we show that $\cC$ can always correct $t=\floor{(d-1)/2}$ errors.
Specifically, let $\vc\in\cC$ and let $\hat{\vc}$ be a word over $\Sigma$ such that $d(\vc,\hat{\vc})\le t$. 
Suppose that $\vc=\Psi(\va,\vb)$ and $\hat{\vc}=\Psi(\hat{a},\hat{b})$.
Then $d(\va,\hat{\va})\le t$ and $d(\vb,\hat{\vb})\le t$.
Since $\vb$ belongs to $\cB$ an $(n,d)_2$-code, we correct the errors in $\hat{\vb}$ to recover $\vb$.

Suppose that $\vb=\phi(i,j,\vp)$. Then we have that $\va\vp$ is a codeword in $\cA$.
Since $\cA$ an $[n+p,n,d]_2$-code, we correct the errors in $\hat{\va}\vp$ to recover $\va\vp$ and 
hence, recover $\va$. Therefore, $\cC$ is an $(n,d)_2$-code.
\end{proof}

\begin{corollary}\label{cor:gc-balanced}
Fix $d$ and set $t=\ceil{(d-1)/2}$. 
There exists an \GC-balanced $(n,d)_4$-code with redundancy at most $(2t+1)\ceil{\log_4 n}+2t$ symbols for sufficiently large $n$. 
%{\color{red} Why???}
\end{corollary}

\begin{proof}
For sufficiently large $n$, we choose an $[n+p,n,d]_2$- and an $[n,k,d]_2$-linear code so that
$p\le t\ceil{log_2 n}+t$ and $n-k\le t\ceil{log_2 n}+t$.
Then applying Construction~\ref{constr-gc-bal}, 
we obtain a \GC-balanced code with at most $(2t+1)\ceil{\log_4 n}+2t$ redundant symbols.
\end{proof}

%{\color{red} Full version: A table of the possible \GC-balanced ECC at length 50?}

\section{Primer Codes}
\label{sec:primer}

In this section, we provide three constructions of primer codes: 
one direct modification of Yazdi \etal{} that yields primer codes for general parameters 
and the other two that rely on cyclic codes and have lower redundancy for a specific set of parameters.

%\subsection{Modification of A Construction by Yazdi et al.}
\subsection{$\kappa$-Mutually Uncorrelated Codes that Avoid Primer Dimer Byproducts of Length $f$}
Yazdi \etal \cite{Yazdietal2018.mu} constructed a set of mutually uncorrelated primers 
that avoids primer dimer byproducts.

\begin{definition}A code $\cA\subseteq \F_2^n$ is {\em $\ell$-APD-constrained}
if for each $\va\in\cA$, 
\begin{itemize}
\item $\va$ ends with one,
\item $\va$ contains $01^\ell0$ as a substring exactly once,
\item $\va$ does not contain $0^\ell$ as a substring.
\end{itemize}
\end{definition}

\begin{lemma}[{Yazdi \etal \cite[Lemma 5]{Yazdietal2018.mu}}]\label{constr-yazdi-primer}
Let $n$, $f$, $\ell$, $r$ be positive integers such that $n=rf+\ell+1$ and $\ell+3\leq f$.
Suppose that $\cA$ is an $\ell$-APD-constrained code of length $f$. Then the code
\[ \C=\{0^\ell1\va_1\va_2\ldots \va_r: \va\in \cA^r\} \]
is both MU and $(2f)$-APD and its size is $|\cA|^r$.
\end{lemma}

The following construction equips the primer code in Lemma~\ref{constr-yazdi-primer} with error-correcting capabilities.
\begin{construction} \label{constr-primer1}
Let $f,r,d$ and $\ell$ be positive integers where $\ell+3\leq f$ and $p+\floor{p/(\ell-1)}+1\le f$. 

{\sc Input}: An $[rf+p,rf,d]_2$-linear code $\cB$ and

\hspace{12mm}an $\ell$-APD-constrained code $\cA$ of length $f$.\\
{\sc Output}: An $(n,d;1,2f)$-primer code $\cC$ of length $n=rf+p+\floor{p/(\ell-1)}+\ell+2$
and size $|\cA|^r$.
\begin{itemize}
\item Consider a systematic encoder for $\cB$. 
\item For every message $\va\in\F_2^{rf}$, let $\va\vp_\va$ be the corresponding codeword in $\cB$.
\item For the vector $\vp_\va$, we insert a one after every $(\ell-1)$ bits and append a one. 
In other words, we insert $\floor{p/(\ell-1)}+1$ ones
and we call the resulting vector $\vp_\va'$.
\item Set $\cC\triangleq\{ 01^\ell\va\vp_\va'\, :\, \va\in\cA^r \}$.
\end{itemize}
\end{construction}

Next, for fixed values of $r$ and $d$, 
we describe a family of $(n,d; 1,f)_2$-primer codes with $n=rf+o(f)$ and 
redundancy at most $t\log_2 n+O(1)$, where $t=\ceil{(d-1)/2}$.
Specifically, we provide the constructions for the input codes $\cB$ and $\cA$ in Construction~\ref{constr-primer1}.

\begin{lemma}For $\ell\ge 8$, set $f= 2^{\ell-4}$.
Then there exists an $\ell$-APD-constrained code $\cA$ of size $(f-\ell-2)2^{f-\ell-4}$.
Furthermore, there is a linear-time encoding algorithm that maps $[f-\ell-2]\times \F_2^{f-\ell-4}$ to $\cA$. 
\end{lemma}

\begin{proof}
We first construct a code $\cA_0$ of length $f-\ell-3$, 
where all codewords do not contain either $0^{\ell-1}$ or $1^{\ell-1}$ as substrings. 
Then $\cA$ can be constructed by inserting $01^\ell0$ to the codewords in $\cA_0$ and appending a symbol $1$.
Since there are $f-\ell-2$ possible positions to insert $01^\ell0$, we have $|\cA|=(f-\ell-2)|\cA_0|$. 

To construct the code $\cA_0$, we use the encoding algorithm $\phi$ proposed by Schoeny \etal \cite{Schoenyetal2017}
that maps a binary sequence of length $(f-\ell-4)$ to a binary sequence of length $(f-\ell-3)$
that avoids $0^{\ell-1}$ and $1^{\ell-1}$ as substrings.
Furthermore, the encoding map $\phi$ has running time $O(f)$.
\end{proof}

Hence, for $\ell\ge 8$, we choose $f=2^{\ell-4}$. 
For the input code $\cB$, we shorten an appropriate BCH code given in Theorem~\ref{primitiveBCH}
to obtain an $[rf+p,rf,d]$-linear code with redundancy $p\le t\log_2 n + t$. 
Hence, applying Construction~\ref{constr-primer1},
we obtain a primer code with  $(n,d; 1,f)_2$-primer codes with $n=rf+p+\floor{p/(\ell-1)}+\ell+2$.

Observe that for sufficiently large $\ell$, we have that $rf< n<(r+1)f$. By choice of $\ell$, we have that $\log_2n + C_1 \le  \ell\le \log_2n + C_2$ for some constants $C_1$, $C_2$ dependent only on $r$.

To analyse the redundancy of the construction, we have that 
\begin{align*}
\log_2 |\cA|^r & = r (f-\ell-4) + r \log_2(f-\ell-2) \\
& \ge r(f-\ell-4)+r \log_2(f/2) \\
&= r(f-\ell-4) + r(\ell-5) = rf - 9r. 
\end{align*}

Therefore, the redundancy is given by $n-\log_2 |\cA|^r$, which is at most
{\small
\begin{align*}
& p+\floor{p/(\ell-1)}+\ell+2+9r \\
&\hspace{3mm} \le (t\log_2n + t) + \frac{t \log_2 n + t}{\log_2 n + C_1-1} + (\log_2 n +C_2) +2+9r\\
&\hspace{3mm} = (t+1) \log_2 n + O(1).
\end{align*}
}
In summary, we have the following theorem.
\begin{theorem}
Fix $r$ and $d$. Then there exists a family of $(n,d;1,f)_2$-primer codes with $n=rf+o(f)$ and
 redundancy at most $(t+1)\log_2 n +O(1)$, where $t=\ceil{(d-1)/2}$. 
Furthermore, there exists a linear-time encoding algorithm for these primer codes. 
\end{theorem}

 Applying Lemma~\ref{coupling}, we obtain primer codes over $\{\A,\T,\C,\G\}$.
 
\begin{corollary}\label{cor:primer-1}
Fix $r$ and $d$, and set $t=\ceil{(d-1)/2}$.
\begin{enumerate}[(i)]
\item There exists a family of $(n,d;1,f)_4$-primer codes with $n=rf+o(f)$ and
 redundancy at most {$(2t+1) \log_4 n +O(1)$}. 
\item There exists a family of balanced $(n,d;1,f)_4$-primer codes with $n=rf+o(f)$ and
 redundancy at most  $(d+1)\log_4 n +O(1)$.
\end{enumerate}
\end{corollary}

\subsection{Almost \GC-Balanced $\kappa$-Mutually Uncorrelated Only}

Using cyclic codes and modifying Construction~\ref{constr-bin-bal}, 
we obtain {\em almost balanced} primer codes that satisfy conditions (P1) and (P2) only.
Here, a code is {\em almost balanced} if the weight (or \GC-content) of every word belongs to 
$\{\floor{n/2}-1,\floor{n/2}, \ceil{n/2}, \ceil{n/2}+1\}$.

Let $n$ be odd and we abuse notation by using $\phi$ to also denote the map
 $\phi:\F_4^n\to \F_4^n$ where $\phi(\va)=\va+\omega^{(n+1)/2}0^{(n-1/2)}$.
 In other words, $\phi$ switches $\A$ with $\C$ and $\T$ with $\G$, and vice versa, in the first $(n+1)/2$ coordinates of $\va$.
%where $\phi(\va)$ is the word obtained by adding $\omega$ to the first $\ceil{n/2}$ coordinates of $\va$.
We have the following analogue of Lemma~\ref{lemma-flip}.

\begin{lemma}\label{lemma-flip-gc}
For $\va\in \F_4^n$, 
we can find  $i\in \llbracket n \rrbracket$ such that  $\phi(\shift^i(\va))$ is \GC-balanced.
\end{lemma}

\begin{construction}\label{constr-primer-2}
Let $n$ be odd, $k\le \ceil{(n+1)/4}$ and $q\in\{2,4\}$\hfill

{\sc Input}: An $[n,k,d]_q$-cyclic code $\cB$ containing  $1^n$.\\
{\sc Output}: An almost balanced $(n,d;k+1,n)_q$-primer code $\cB$ of size at least $q^k/n$.

\begin{itemize}
\item Let $\vu_1, \vu_2, \ldots, \vu_m$ be the set of representatives $\cC/\equivcyclic$. 
\item For each $\vu_i$, find $j_i\in [n]$ such that 
$\phi(\shift^{j_i}(\vu_i))$ is either balanced or \GC-balanced.
\item Let $\mu=(n-1)/2$. For each $\vu_i$, set 
\[ \vv_i =
\begin{cases}
\shift^{j_i}(\vu_i)+1^{\mu+1}0^{\mu-1}1, &\mbox{ if $q=2$}, \\
\shift^{j_i}(\vu_i)+\omega^{\mu+1}0^{\mu-1}\omega, &\mbox{ if $q=4$}.
\end{cases}\]
\item Set $\cC=\{\vv_i:1\leq i\leq m\}$.
\end{itemize}
%Let $\cC'=\{\vv_i:1\leq i\leq m\}$. 
\end{construction}

\begin{theorem}\label{thm-primer-2}
Construction~\ref{constr-primer-2} is correct. 
In other words,  $\cC$ is an almost balanced $(n,d;k+1,n)_q$-primer code of size at least $q^k/n$.
\end{theorem}

To prove Theorem~\ref{thm-primer-2}, we require the following technical lemma modified from Yazdi {\em et al.} \cite{Yazdietal2018.mu}.

\begin{lemma}\label{lemma-0run}
Let $\cC$ be a cyclic code of dimension $k$ containing $1^n$.
Then the run of any symbols in any non-constant codeword is at most $k-1$.
\end{lemma}

\begin{proof}[Proof of Theorem~\ref{thm-primer-2}]
Since $\cC$ is coset of $\cB$, we have that $\cC$ is an $(n,d)_q$-code. 
For $i\in [m]$, since $\phi(\shift^{j_i}(\vu_i))$ is balanced and $\vv_i$ differs from the former in one symbol, we have that $\vv_i$ is almost balanced.

Now, we demonstrated weakly mutually uncorrelatedness for the case of $q = 2$. 
The case of $q = 4$ can be proceeded in the same way. 
Suppose on the contrary that $\cC$ is not $k$-WMU.
Then there is a proper prefix $\vp$ of length $\ell$, $\ell \ge k+1$ such that both $\vp\va$ and $\vb\vp$ belong to $\cC$. In other words, $\cB$ contains the words
\[ \vp\va + 1^{\mu+1}0^{\mu-1}1\mbox{ and }\vb\vp + 1^{\mu+1} 0^{\mu-1}1,\]
where $\mu=(n-1)/2$.  
Consequently, since $\cB$ is cyclic, we have that $\vp\vb+\shift^\ell(1^{\mu+1} 0^{\mu-1}1)$ belongs
to $\cB$. Hence, by linearity of $\cB$, the word
\[\vc\triangleq 0^\ell(\va - \vb) + 1^{\mu+1} 0^{\mu-1}1+\shift^\ell(1^{\mu+1} 0^{\mu-1}1)\]
belongs to $\cB$. We look at prefix of length $\ell$ of $\vc$.
\begin{itemize}
\item When $\ell\le \mu$, the word $\vc$ has prefix $1^{\ell-1}0$.
Hence, $\vc$ is a non-constant codeword of $\cC$ and since $\ell-1\ge k$, this contradicts Lemma~\ref{lemma-0run}.

\item When $\ell=\mu+1$, the word $\vc$ has prefix $01^{\mu-1}$. %, respectively. 
Hence, $\vc$ is a non-constant codeword of $\cC$ and
since $\mu-1\ge k$, this contradicts Lemma~\ref{lemma-0run}.

\item When $\ell\ge \mu+2$, the word $\vc$ has prefix $0^{\ell-\mu}1^{2\mu+1-\ell}$.
Since either $\ell-\mu$ or $2\mu+1-\ell$ is at least $\ceil{\mu+1/2}=\ceil{(n+1)/4}\ge k$,
the word $\vc$ contains a run of ones or zeros of length $k$, contradicting Lemma~\ref{lemma-0run}.
\qedhere
\end{itemize}
\end{proof}

%\begin{example}
%Let $\cC$ be an $[15,9,5]$-cyclic code with generator polynomial 
%$x^6+x^5+(\omega+1)x^4+x^3+(\omega+1)x^2+x+1$.
%Then Construction~\ref{constr-primer-2} provides a balanced $(15,5; 9,15)_4$-primer code $\cB$ of size $4^9/15\ge 4^7=2^{14}$.
%
%In contrast, for their experiement, 
%Yazdi \etal{} constructed a set of weakly mutually uncorrelated primers of length 16, distance four and size four. 
%Specifically, they set $\cC_1=\{01^701^7,10^710^7\}$ and $\cC_2$ to be an extended BCH $[16,11,4]$-cyclic code.
%Then they applied the coupling construction to obtain an $(16,4; 9,9)$-primer code of size $2^{12}$.
%
%Therefore, Construction~\ref{constr-primer-2} provides a larger set of primers using less bases, 
%while improving the minimum distance and enforcing \GC-balanced constraint on all codewords.
%\end{example}

\subsection{$\kappa$-Mutually Uncorrelated Codes that Avoid Primer Dimer Byproducts of Length $\kappa$}

%{\color{blue}
Using reversible cyclic codes, we further reduce the redundancy for primer codes 
in the case when $\kappa=f$.

\begin{definition}\label{def:rc-generating}
Let $g(X)$ be the generator polynomial of a reversible cyclic code $\cB$ of length $n$ and dimension $k$ that contains $1^n$.
Set $h(X)=(X^n-1)/g(X)$. 
The set $\{h^*(X), p_1(X),p_2(X),\ldots, p_P(X)\}$ of polynomials is {\em $(g,k)$-rc-generating}
if the following hold:
\begin{enumerate}[(R1)]
\item $h^*(X)$ divides $h(X)$;
\item $h^*(1)\ne 0$;
\item $h^*(X)=X^{d^*}h^*(X^{-1})/h^*(0)$, where $d^*=\deg h^*$;
\item $h^*(X)$ does not divide $X^s p_i(X) - p_j(X)$ for all $i,j\in [P]$ and $s\in [n-1]$.
\item $h^*(X)$ does not divide $X^s p_i(X) - X^{k-1} p_j(X^{-1})$ for all $i,j\in [P]$ and $0\le s\le{n-k}$.
\item $h^*(X)$ does not divide $X^{s+k-1} p_i(X^{-1}) - p_j(X)$ for all $i,j\in [P]$ and $0\le s\le{n-k}$.
\item $\deg p_i(X) < \deg h^*$ for $i\in [P]$. 
\end{enumerate}
\end{definition}

\begin{construction}\label{constr-primer-3}
\hfill

{\sc Input}: An $[n,k,d]_q$-reversible cyclic code $\cB$ containing $1^n$ with generator polynomial $g(X)$
 and a $(g,k)$-rc-generating set of polynomials $\{h^*(X), p_1(X),p_2(X),\ldots, p_P(X)\}$ .\\
{\sc Output}: An $(n,d;k,k)_q$-primer code $\cC$ of size $q^{k^*}P$, where
$k^* = k - \deg h^*$.
\begin{itemize}
\item Set \[\cC \triangleq \{ (m(X)h^*(X)+p_i(X))g(X): \deg m < k^*, i\in [P]\}.\]
\end{itemize}
\end{construction}

\begin{theorem}\label{thm-primer-3}
Construction~\ref{constr-primer-3} is correct. 
In other words, $\cC$ is an $(n,d;k,k)_q$-primer code.
\end{theorem}

We illustrate Construction~\ref{constr-primer-3} via an example.

\begin{example}\label{exa:primer}
Set $n=15$ and $q=4$.
Let $g(x)=x^6 + x^5 + (\omega + 1)x^4 + x^3 + (\omega + 1)x^2 + x + 1$ be the generator polynomial of an $[15,9,5]_4$-reversible cyclic code that contains $1^n$.
Consider $h^*(X) =X^4 + \omega X^3 + \omega X^2 + \omega X + 1$ and
\vspace{-3mm}

{\scriptsize
\begin{align*}
p_1 & =\omega, & 
p_{10} &= \omega x^3 + (\omega + 1)x^2 + x + \omega + 1, \\
p_2 & =\omega + 1, &
 p_{11} &= \omega x^3 + (\omega + 1)x^2 + x + 1, \\
p_3 & =1, 
& p_{12} &= \omega x^3 + x^2, \\
p_4 & =\omega x + \omega, 
& p_{13} &= \omega x^3 + x^2 + x + 1, \\
p_5 & =(\omega + 1)x + \omega + 1, & 
p_{14} &= (\omega + 1) x^3 + \omega x^2 + \omega x + \omega, \\
p_6 & =x + 1, & 
p_{15} &= (\omega + 1)x^3 + \omega x^2 + (\omega + 1)x + \omega + 1, \\
p_7 & =\omega x^2 + \omega x + \omega + 1, 
& p_{16} &= (\omega + 1)x^3 + x^2 + \omega x + 1, \\
p_8 & =\omega x^3 + (\omega + 1) x^2 + 1, & 
p_{17} &= x^3 + \omega x^2 + (\omega + 1)x + \omega + 1.\\
p_9 & =\omega x^3 + (\omega + 1)x^2 + x,
\end{align*}
} 
\vspace{-5mm}

We can verify that the set $\{h^*(X),p_1(X),\ldots, p_{17}(X)\}$ is $(g,9)$-rc-generating.
Therefore, $k^*=15-6-4=5$ and the size of the $(15,5;9,9)_4$-primer code have size
$17(4^5)\ge 2^{14}$.

In contrast, for their experiment, Yazdi \etal{} constructed a set of weakly mutually uncorrelated primers of length 16, distance four and size four. 
Specifically, they set $\cC_1=\{01^701^7,10^710^7\}$ and $\cC_2$ to be an extended BCH $[16,11,4]$-cyclic code.
Then they applied the coupling construction to obtain an $(16,4; 9,16)$-primer code of size $2^{12}$.

Therefore, Construction~\ref{constr-primer-3} provides a larger set of primers using less bases, 
while improving the minimum distance and avoiding primer dimer products at the same time.
\end{example}

We outline our steps in establishing Theorem~\ref{thm-primer-3}. 
First, we demonstrate Lemma~\ref{lem:combinatorial-rc}. 
The lemma provides certain {\em combinatorial} sufficiency conditions 
for a subcode of a reversible cyclic code to be a primer code.
Next, using the {\em algebraic} properties of the polynomials in Construction~\ref{constr-primer-3},
we then show that $\cC$ satisfy the combinatorial conditions in Lemma~\ref{lem:combinatorial-rc}. 
The second step is deferred to Section~\ref{sec:encoding}.

\begin{lemma}\label{lem:combinatorial-rc}
Let $\cB$ be an $(n,k,d)_q$ reversible cyclic code containing $1^n$.
Let $\cC \subseteq \cB$ be a subcode such that for any two codewords $\vu,\vv$ in $\cC$, not necessarily distinct, the following holds.
\begin{enumerate}[(S1)]
\item $\shift^i(\vu)\not = \vv$ for  $k\leq i <n$;
\item  $\shift^i(\vu)\not = \overline{\vv}$ for $0\leq i \leq n-k$;
\item  $\shift^i(\vu)\not = \vv^{rc}$ and $\shift^i(\vu^{rc})\not=\vv $ for $0\leq i \leq n-k$.
\end{enumerate}
Then $\cC$ is an $(n,d;k,k)_q$-primer code.
\end{lemma}

\begin{proof}
Since $\cC$ is a subcode of $\cB$, we have that $\cC$ is an $(n,d)_q$-code. 
It remains  to show the WMU and APD properties. 

We first show that $\cC$ is $k$-WMU. Suppose to the contrary that there is a proper sequence $\vp$ of length $\ell$, where $k \leq \ell < n$, such that both $\vp\va$
and $\vb\vp$ belong to $\cC$. Since  $\cC\subseteq \cB$ and  $\cB$ is a cyclic code,  the word  $\vp\va-\vp\vb=0^\ell(\va-\vb)$ belongs to $\cB$.
Since $\ell\ge k$, Lemma~\ref{lemma-0run} implies that $\va=\vb$ and so, 
\begin{align*}
\shift^\ell(\vp\va)=\va\vp=\vb\vp.
\end{align*}
Since $k\leq \ell<n$ and $\vp\va$ and $\vb\vp$ belong to  $\cC$, we obtain a contradiction for condition (S1). 

Now we show that $\cC$ is a $k$-APD code. 
Towards a contradiction, we suppose that there is a proper sequence $\vp$ of length $\ell$, where $k \leq \ell < n$, such that both  $\va_1\vp \vb_1$ and $\va_2\overline{\vp}\vb_2$ belong to $\cC$. 
Since $\cC\subseteq \cB$ and $\cB$ is a cyclic code containing $1^n$,  the word $\vp\vb_1\va_1-\vp\overline{\vb_2}\overline{\va_2}= \mathbf{0} (\vb_1\va_1-\overline{\vb_2}\overline{\va_2})$ also belongs to $\cB$. It follows from Lemma~\ref{lemma-0run} that
\[ \vp\vb_1\va_1=\vp\overline{\vb_2}\overline{\va_2}.\]
%Write the length of a sequence $\vx$ as $|\vx|$.  
Without loss of generality, we assume that $|\va_2| \geq  |\va_1|$. Then
\[
\shift^{\ell'}\left({\va_1\vp \vb_1}\right)=\overline{\va_2}{\vp}\overline{\vb_2}=\overline{\va_2 \overline{\vp} \vb_2},
\]
where $\ell'=|\va_2|-|\va_1|$. % and the first equality holds as $\vb_1\va_1=\overline{\vb_2}\overline{\va_2}$. 
Since the length of $\vp$ is no less than $k$, we have that $\ell'\leq n-k$, 
which contradicts condition (S2).
%that both  $\va_1\vp \vb_1$ and $\va_2\overline{\vp}\vb_2$ belong to $\cC$.

Finally, suppose that  $\va_1\vp \vb_1$ and $\va_2\vp^{rc}\vb_2$ belong to $\cC$, where $\vp$ is a proper sequence of length $\ell$ and  $k \leq \ell <n$. Proceeding as before, we can show that 
\[\vp\vb_1\va_1=\vp\va_2^{rc}\vb_2^{rc},\]
or equivalently,
\[\vb_1\va_1=\va_2^{rc}\vb_2^{rc} \textrm{ and  } \va_1^{rc}\vb_1^{rc}=\vb_2\va_2.\]

If $ |\vb_2|\geq |\va_1|$, we have that 
\begin{align*}
\shift^{\ell'}\left({\va_1\vp \vb_1}\right)=\vb_2^{rc}\vp\va_2^{rc}=(\va_2\vp^{rc}\vb_2)^{rc},
\end{align*}
where $\ell'= |{\vb}_2|- |\va_1| \leq n-k$,  contradicting the first inequality of Condition (S3); if $|\vb_2| < |\va_1|$, then $|\va_2| > |\vb_1|$ and we have 
\begin{align*}
\shift^{\ell'}\left(({\va_1\vp \vb_1})^{rc}\right)= \shift^{\ell'}\left(\vb_1^{rc}\vp^{rc} \va_1^{rc} \right) = \va_2\vp^{rc}\vb_2,
\end{align*}
where $\ell'  = |\va_2| - |\vb_1|  \leq n-k$,  contradicting the second inequality of Condition (S3). 
%{\color{red} Since these two cases lead to two  conditions, we'd better not omit the second one.}
\end{proof}

Finally, applying Construction~\ref{constr-primer-3} to the class of reversible cyclic codes in Theorem~\ref{thm-LCDwithone}, we obtain a family of primer codes that has efficient encoding algorithms. 
The detailed proof is deferred to Section~\ref{sec:encoding}.

\begin{corollary}\label{cor-family-primer}
Let $m\ge 6$ and $1\le \tau\le \ceil{m/2}$
Set $n=4^m-1$ and $d=4^\tau-1$.
There exists an $(n,d;k,k)_4$-primer code of 
%size ${4^{k-1}}/{(n+1)^2}$, 
size $4^{k-2m}$, 
where 
\begin{equation*}
k = 
\begin{cases} n-(d-3)m, & \mbox{if $m$ is odd and $\tau=\frac{m+1}{2}$;} \\
n-(d-1)m,  & \mbox{otherwise.}\\
\end{cases}
\end{equation*}
Therefore, there is a family of  $(n,d;k,k)_4$-primer codes with 
$d\approx \sqrt{n}$, $k\approx n-\sqrt{n}\log_4 n$, and redundancy at most $(d+1)\log_4(n+1)$.
\end{corollary}

\section{Codes for DNA Computing}
 \label{sec:DNAcomputing}
 
 {%\color{blue}
 Since Adleman demonstrated the use of DNA hybridization to solve a specific instance of the directed Hamiltonian path problem \cite{Adleman:1994}, 
 the coding community have investigated the possibility of error control via code design 
\cite{Marathe:2001,Milenkovic:2005}.
 %Specifically, the following coding constraints were introduced.
 In this paper, we focus on designing codes with the following constraints.
 
 \begin{definition}
A \GC-balanced $(n,d)_4$-code is a balanced $(n,d)$-{\em DNA computing code} if the following hold. 
\begin{enumerate}[(C1)]
\item $d(\va,\vb^r)\geq d$ for all $\va,\vb\in \cC$.
\item  $d(\va,\vb^{rc})\geq d$ for all $\va,\vb\in \cC$.
%\item For all $\va\in \cC$, its {\it \GC-content}, the number of symbols that correspond to either $\G$ or $\C$, is equal to $w$.  
\end{enumerate}
%When $w=n/2$, we refer to the code as a {\em balanced} $(n,d)$-{DNA computing code}.
\end{definition}

More generally, DNA computing codes require that the \GC-content, the number of symbols that correspond to either $\G$ or $\C$, of all codewords to be the same or approximately the same.
As always, the fundamental problem for DNA computing codes is to find the largest possible codes satisfying the constraints above. Many approaches have been considered for this problem. 
These include search algorithms, template-based constructions and constructions over certain algebraic rings
(see, Limbachiya \etal{} \cite{limbachiyaetal2016} for a survey). 
%The known DNA codes usually satisfy two or three of the constraints above. 
%A survey of the literature in this area  has been undertaken by  Limbachiya et al. \cite{limbachiyaetal2016}.

There are few explicit families of DNA computing codes satisfying all constraints for large $n$.  
In this section we propose a class of balanced DNA computing codes that satisfies both the constraints (C1) and (C2).
 
 We modify our balancing techniques in Sections~\ref{sec:balanced} and~\ref{sec:primer}.
 Recall that by flipping, we mean exchanging $\A$ with $\C$ and $\T$ with $\G$. 
 Then Lemma~\ref{lemma-flip-gc} states that for $\va\in \F_4^n$, 
 we can balance one of its cyclic shifts by flipping its first $\ceil{n/2}$ components. 
 However, in order to accommodate the reverse and reverse-complement distance constraints, 
we do the following. 

Let $n$ be odd and set $s$ be the integer nearest to $n/4$.
In other words, $s$ is the unique integer in the set $\{(n-1)/4, (n+1)/4\}$. 
Let $\pi: \F_4^n\to \F_4^n$ be the map such that $\pi(\va)=\va+\omega^s0^{n-2s}\omega^s$
for any $\va \in  \F_4^n$.
In other words, $\pi$  flips the first $s$ and the last $s$ symbols of $\va$.  
The following lemma follows directly from Lemma~\ref{lemma-flip-gc}.
  
\begin{lemma}
Let $n$ be  odd. For any $\va\in\F_4^n$, there exists  $i\in \bbracket{n}$ 
such that  $\pi(\shift^i(\va))$  is \GC-balanced. %content in $\{(n-1)/2 , (n+1)/2\}$. 
\end{lemma}  

As before, we next define a set of polynomials that enables us to generate our code efficiently.

\begin{definition}\label{def:rc2-generating}
Let $g(X)$ be the generator polynomial of a reversible cyclic code $\cB$ of length $n$ and dimension $k$ that contains $1^n$.
Set $h(X)=(X^n-1)/g(X)$. 
The set $\{h^*(X), p_1(X),p_2(X),\ldots, p_P(X)\}$ of polynomials is {\em $(g,k)$-rc2-generating}
if the set obeys conditions (R1) to (R4), (R7) in Definition~\ref{def:rc-generating} and 
\begin{enumerate}[(R5')]
\item $h^*(X)$ does not divide $X^s p_i(X) - X^{k-1} p_j(X^{-1})$ for all $i,j\in [P]$ and $s\in\bbracket{n}$.
\end{enumerate}
\end{definition}

It is immediate from definition that an $(g,k)$-rc-generating set is also an $(g,k)$-rc2-generating set.

 \begin{construction}\label{constr-DNA}
 Let $n$ be odd.
\hfill

{\sc Input}: An $[n,k,d]_q$-reversible cyclic code $\cB$ containing $1^n$ with generator polynomial $g(X)$
 and a $(g,k)$-rc2-generating set of polynomials $\{h^*(X), p_1(X),p_2(X),\ldots, p_P(X)\}$ .\\
{\sc Output}: A balanced $(n,d)_4$-DNA computing code $\cC$ of size $4^{k^*}P$, where
$k^* = k - \deg h^*$.
\begin{itemize}
\item Set \[\cA \triangleq \{ (m(X)h^*(X)+p_i(X))g(X): \deg m < k^*, i\in [P]\}.\]
\item For $\vu\in \cA$, find $i_\vu\in \bbracket{n-1}$ such that 
$\vv_\vu=\pi(\shift^{i_\vu}(\vu))$ is \GC-balanced.
\item Set $\cC=\{\vv_\vu:  \vu\in \cA \}$.
\end{itemize}
\end{construction}
 
 \begin{theorem}\label{thm-DNA}
Construction~\ref{constr-DNA} is correct. 
In other words, $\cC$ is a balanced $(n,d)_4$-DNA computing code of size $4^{k^*}$.
\end{theorem}

%We illustrate Construction~\ref{constr-primer-3} via an example.
%
%\begin{example}\label{exa:DNA}
%Set $n=15$.
%Let $g(x)=x^6 + x^5 + (\omega + 1)x^4 + x^3 + (\omega + 1)x^2 + x + 1$ be the generator polynomial of an $[15,9,5]_4$-reversible cyclic code that contains $1^n$.
%
%Consider the set $\{h^*(X),p_1(X),\ldots, p_{17}(X)\}$ given in Example~\ref{exa:primer}.
%Since the set is $(g,9)$-rc-generating, it is also $(g,9)$-rc2-generating.
%Therefore, we can apply Construction~\ref{constr-DNA} to obtain a balanced $(15,5)_4$-DNA computing code of size
%$17(4^5)\approx 2^{14.08}$.
%
%On the other hand, Gaborit and King constructed an $(15,5)_4$-DNA computing code of size $25670\approx 2^{14.64}$ with \GC-content seven~\cite{Gaborit:2015}. 
%Gaborit and King's method involves taking a subcode of a certain linear code and using its \GC-weight enumerator to compute its size. Hence, it is unclear how to encode messages into this subcode. 
%
%Therefore, Construction~\ref{constr-DNA} provides a DNA computing code of comparable size and 
%in addition, describes an efficient method to encode the messages.
%\end{example}
 
As in Section~\ref{sec:primer}, to provide Theorem~\ref{thm-DNA},
we first provide certain combinatorial sufficiency conditions 
for a subcode of a reversible cyclic code to be a DNA computing code,
and then show that $\cC$ satisfy these combinatorial conditions.
As before, we defer the second step to Section~\ref{sec:encoding}.

\begin{lemma}\label{lem:combinatorial-rc2}
Let $\cB$ be an $(n,k,d)_q$ reversible cyclic code containing $1^n$.
Let $\cA \subseteq \cB$ be a subcode such that for any two codewords $\vu,\vv$ in $\cA$, not necessarily distinct, the following holds.
\begin{enumerate}[(S1')]
\item $\shift^i(\vu)\not = \vv$ for  $i\in [n-1]$;
\item $\shift^i(\vu) \ne \vv^r$ for $i \in \bbracket{n}$;
\item $\shift^i(\vu) \ne \vv^{rc}$ for $i \in \bbracket{n}$.
\end{enumerate}
If we define $\cC$ as in Construction~\ref{constr-DNA},
then  $\cC$ is a balanced $(n,d)_4$-DNA computing code of size $|\cA|$.
\end{lemma}

\begin{proof}
First, condition (S1') ensures that the codewords $\vv_\vu$ and $\vv_\vu'$ are distinct whenever $\vu\ne\vu'$.
Therefore, the size of $\cC$ is given by $|\cA|$.

Next, by choice of $i_\vu$, we have that all codewords in $\cC$ are \GC-balanced. 
Since $\cC$ belongs to a coset of $\cB$, we have that $\cC$ is an $(n,d)_4$-code.

Therefore, it remains to demonstrate constraints (C1) and (C2).
For any $\va, \vb\in \cC$, let $\vu, \vv$ be the corresponding vectors in $\cA$. In other words,
\[\va =\pi(\shift^{i_\vu}(\vu)) \mbox{ and }\vb =\pi(\shift^{i_\vv}(\vv)).\]

We first show that $d(\va, \vb^r)\ge d$.
Since $\cA$ satisfies condition (S2'), we have that $\shift^{i_\vu}(\vu)\ne \shift^{i_\vv}(\vv)^r$.
Since $\shift^{i_\vu}(\vu)$, $\shift^{i_\vv}(\vv)^r$ belongs to $\cB$, 
we have that $d(\shift^{i_\vu}(\vu), \shift^{i_\vv}(\vv)^r)\geq d$. 
Now, $\omega^s0^{n-2s}\omega^s=(\omega^s0^{n-2s}\omega^s)^r$, and
so, $\vb^r=\pi(\shift^{i_\vv}(\vv))^r=\pi(\shift^{i_\vv}(\vv)^r)$. Therefore,
\[ d(\va, \vb^r) = d(\pi(\shift^{i_\vu}(\vu)), \pi(\shift^{i_\vv}(\vv)^r))=d((\shift^{i_\vu}(\vu), \shift^{i_\vv}(\vv)^r))\ge d.\]

Constraint (C2) can be similarly demonstrated.
\end{proof}

As before, we apply Construction~\ref{constr-DNA} to the reversible cyclic codes in Theorem~\ref{thm-LCDwithone} to obtain a family of balanced DNA computing codes. 
The proof is deferred to Section~\ref{sec:encoding}.

 \begin{corollary}\label{cor-family-computing}
Let $m\ge 6$ and $1\le \tau\le \ceil{m/2}$
Set $n=4^m-1$, $d=4^\tau-1$, and 
\begin{equation*}
k = 
\begin{cases} n-(d-3)m, & \mbox{if $m$ is odd and $\tau=\frac{m+1}{2}$;} \\
n-(d-1)m,  & \mbox{otherwise.}\\
\end{cases}
\end{equation*}

Then there exists a \GC-balanced $(n,d)_4$-DNA computing code of 
%size ${4^{k-1}}/{(n+1)^2}$, 
size at least $4^{k-2m}$. %whose codewords can be encoded efficiently. 
Therefore, there exists a family of \GC-balanced  $(n,d)_4$-primer codes with 
$d\approx \sqrt{n}$ and redundancy at most $(d+1)\log_4 (n+1)$.
Furthermore, these codes have efficient encoding algorithms.
\end{corollary}

%\section{Efficient Encoding Into Cyclic Equivalence Classes}\label{sec-encode-equiv}
\section{Efficient Encoding Into Cyclic Classes}\label{sec:encoding}
%{\color{red} TO NAME THIS SECTION}

In this section, unless stated otherwise, all words are of length $n$ and 
we index them using $\bbracket{n}$.
Recall that a word $\vc\in\F_q^n$ is identified with the polynomial $c(X)=\sum_{i=0}^{n-1} c_iX^i$.
We further set $X^n=1$ and hence, all polynomials reside in the quotient ring $\F_q[X]/\pspan{X^n-1}$.

Hence, in this quotient ring, we have the following properties. 
Let $c(X)\in \F_q[X]/\pspan{X^n-1}$ be the polynomial corresponding to the word $\vc$.
\begin{itemize}
\item For $s\in\bbracket{n}$, the polynomial $X^s c(X)$ corresponds to the word $\shift^i(\vc)$.
\item  $X^{n-1}c(X^{-1})$ corresponds to the word $\vc^r$. Given $c(X)$, we further define the {\em reciprocal polynomial} of $c(X)$ to be $c^\dagger(X)=X^{\deg c}c(X^{-1})$ and we say $c(X)$ is self-reciprocal if $c(0)\ne 0$ and $c(X)=c^\dagger(X)/c(0)$. 
\item $(X^n-1)/(X-1)$ corresponds to $1^n$, and so, 
$c(X)+(X^n-1)/(X-1)$  corresponds to $\overline{\vc}$.
 \item  $X^{n-1}c(X^{-1})+(X^n-1)/(X-1)$ corresponds to $\vc^{rc}$.
\end{itemize}

From these observations, we can then easily characterise when a cyclic code contains $1^n$ or when a cyclic code is reversible.

\begin{proposition}\label{cyclic}
Let $\cC$ be a cyclic code with generator polynomial $g(X)$. Then
\begin{enumerate}[(i)]
\item $\cC$ contains $1^n$ if and only if $(X-1)$ does not divide $g(X)$, {\em i.e.} $g(1)\ne 0$.
\item $\cC$ is reversible if and only if $g(X)$ is self-reciprocal.
\end{enumerate}
\end{proposition}

Next, we review the method of Tavares \etal{} that efficiently encodes into distinct cyclic classes.
We restate a special case of their method and 
reproduce the proof here as the proof is instructive for the subsequent encoding methods.

\begin{theorem}[Tavares \etal \cite{Tavaresetal1971}]\label{thm-cyclerep}
Let $\cB$ be a cyclic code of dimension $k$ with generator polynomial $g(X)$
and define $h(x)=(X^n-1)/g(X)$.
Suppose $h^*(X)$ divides $h(X)$ and $h(X)$ does not divide $X^s-1$ for $s\in[n-1]$.
Set $k^* = k-\deg h^*(X)$
\[ \cB^* = \{(m(X)h^*(X)+1)g(X) : \deg m < k^*\}.\]
Then $\cB^* \subseteq \cB/\equivcyclic$\, .
\end{theorem}

\begin{proof}
It suffices to show for distinct polynomials $m(X)$ and $m'(X)$ with $\deg m, \deg m' <k^*$ and $s\in[n-1]$, we have that
{\small
\[ X^s(m(X)h^*(X)+1)g(X) \ne (m'(X)h^*(X)+1)g(X) \pmod{X^n-1}. \]
}
To do so, we prove by contradiction and suppose that equality holds. 
In other words, there exists a polynomial $f(X)$ such that 
{\small
\[ X^s(m(X)h^*(X)+1)g(X) = (m'(X)h^*(X)+1)g(X) +f(X)(X^n-1). \]
}
Dividing throughout by $g(X)$ and rearranging the terms, we have that 
\[ (X^s-1)+(X^s m(X)-m'(X))h^*(X) = f(X)h(X).\]
Since $h^*(X)$ divides $h(X)$, then $h^*(X)$ must divide $X^s-1$, yielding a contradiction.
\end{proof}

Suppose $n=2^m-1$ in Theorem~\ref{thm-cyclerep}. 
It is not difficult to show that the degree of $h^*(X)$ is $m$. 
Thus, Theorem~\ref{thm-cyclerep} encodes into $2^{k-m}={2^k}/{(n+1)}$ cyclic classes. 
Since the size of the cyclic code is $2^k$ and each class contains at most $n$ words,
the theorem in fact encodes most cyclic classes. 

The method of Tavares \etal{} \cite{Tavaresetal1971} encodes more classes by considering more factors of $h(X)$ that satisfy the conditions of the theorem.
In some special cases, like $n$ is a prime, this iterative process can encode all the cyclic classes.
%If $h(x)$ has more factors other than $h_1(x)$, say $h_2(x), h_3(x) \ldots,$   of exponent $n$, we can apply the theorem to subcodes $\langle h_1(x)g(x) \rangle, \langle h_1(x) h_2(x)g(x) \rangle, \ldots$ to obtain more cyclic classes, and all these  classes are distinct. In some special cases, for example when $n$ is a prime, this process can enumerate all the cyclic classes.

\begin{table*}[!t]
{\small
\begin{tabular}{|p{0.8cm}|p{5cm}|p{5cm}|p{7cm}|}
\hline
{Constr.} & Input & Output & Redundancy for Infinite Family \\ \hline \hline

A &
binary cyclic code &
balanced binary code & 
$(t+1)\log_2 n+1$, where $t=d/2-1$\newline
({\em c.f.} Corollary~\ref{cor:bin-balanced}) \\\hline

B &
two binary linear codes &
\GC-balanced code&
 $(2t+1)\log_4 n+2t$, where $t=\ceil{(d-1)/2}$\newline
({\em c.f.} Corollary~\ref{cor:gc-balanced}) \\\hline

C &
binary linear code, and \newline
$\ell$-APD-constrained code&
primer code & 
$(2t+1) \log_4 n +O(1)$, where $t=\ceil{(d-1)/2}$ \newline (no \GC-balanced constraint) \newline
$(d+1)\log_4 n +O(1)$ (\GC-balanced)\newline
({\em c.f.} Corollary~\ref{cor:primer-1}) \\\hline

D &
cyclic code containing $1^n$&
almost \GC-balanced $(n,d;\kappa,n)$-primer code &
%\newline  
%with WMU feature only & 
N.A.  \\\hline

E &
reversible cyclic code containing $1^n$, \newline
and 
rc-generating set of polynomials &
primer code with $\kappa=f$ & 
$(d+1)\log_4(n+1)$
({\em c.f.} Corollary~\ref{cor-family-primer}) \\\hline

F &
reversible cyclic code containing $1^n$, \newline
and
rc2-generating set of polynomials &
\GC-balanced DNA computing codes & 
$(d+1)\log_4(n+1)$
({\em c.f.} Corollary~\ref{cor-family-computing}) \\\hline
\end{tabular}
}
\caption{Summary of Constructions for Codes of Length $n$ and Distance $d$}
\label{table:summary}
\end{table*}

\subsection{Detailed Proofs for Section~\ref{sec:primer}}

Borrowing ideas from Tavares \etal, we complete the proof of Theorem~\ref{thm-primer-3}.
Specifically, we demonstrate the following lemma.

\begin{lemma}
Let $g(X)$ be the generator polynomial of a reversible cyclic code $\cB$ of length $n$ and dimension $k$ that contains $1^n$. 
If $\{h(X), p_1(X),  p_2(X), \ldots, p_P(X)\}$ is $(g,k)$-rc-generating and $k^*=k-\deg h^*$,
then the subcode $\cC=\{(m(X)h^*(X)+p_i(X))g(X): \deg m < k^*, \, i\in[P]\}$ satisfies conditions (S1) to (S3) in Lemma~\ref{lem:combinatorial-rc}. 
\end{lemma}

\begin{proof}Here we only prove condition (S3). The other two conditions can be proved similarly. 
In particular, we demonstrate that the violation of condition (S3) contradicts either condition (R5) or condition (R6) in Definition~\ref{def:rc-generating}.

Suppose to the contrary of (S3). We first assume there are two codewords $\vu,\vv\in\cC$ 
such that $\shift^s(\vu)=\vv^{rc}$ for some $s\in \bbracket{n-k}$.
The other case can be treated similarly.
Hence, there exist two polynomials $m(X)$ and $m'(X)$ with degrees strictly less than $k^*$, 
two polynomials $p_i(X)$ and $p_j(X)$ with $i,j\in [P]$
such that the following equality holds with some polynomial $f(X)$.
\begin{align*}
&X^s\left(m(X)h^*(X)+p_i(X)\right)g(X)\\
&\hspace{10mm}= X^{n-1}\left(m'(X^{-1})h^*(X^{-1})+p_j(X^{-1})\right)g(X^{-1})\\
&\hspace{40mm}+\frac{X^n-1}{X-1}+ f(X)(X^n-1).
\end{align*}

Since $\cB$ is a reversible code, $g(x)$ is self-reciprocal, {\em i.e.},
$X^{n-k}g(X^{-1})=g(X)$. Similarly, we have $h^*(X)=X^{\deg h^*}h^*(X^{-1})/h^*(0)$.
Dividing the equation by $g(X)$ and rearranging the terms, we have the following equality.
\begin{align*}
&\left(X^s p_i(X) - X^{k-1}p_j(X^{-1})\right)\\
&\hspace{10mm}+ \left(X^s m(X) - X^{k^*-1}  m'(X^{-1}) h(0)\right)h^*(X) \\
&\hspace{40mm}=\frac{h(X)}{X-1}+ f(X)h(X).
\end{align*}
%\begin{align*}
%&X^s p_i(X) - X^{k-1}p_j(X^{-1})\\
%&\hspace{2mm}= \frac{h(X)}{X-1}+ f(X)h(X)- \left(X^s m(X) - X^{k^*}  m'(X^{-1}) h(0)\right)h^*(X).
%\end{align*}
Since $h^*(1)\ne 0$, we have that $h^*(X)$ divides $h(X)/(X-1)$.
Therefore, 
\[h^*(X)\mbox{ divides } X^s p_i(X) - X^{k-1}p_j(X^{-1}),\]
contradicting condition (R5) in Definition~\ref{def:rc-generating}.
\end{proof}

Next, we complete the proof of Corollary~\ref{cor-family-primer}.
To do so, we recall some concepts in finite field theory.

For $m\ge 2$, we consider the finite field $F\triangleq \F_{4^m}$.
A nonzero element $\alpha\in F$ is said to be {\em primitive} if $\alpha^i\ne 1$ for $i\in [4^m-2]$.
%It is can be shown that the number of primitive elements in $F$ is given by $\varphi(4^m-1)$.
For $\alpha\in F$, we let $M(\alpha)$ denote the {\em minimal polynomial} of $\alpha$ in the base field $\F_4$. Then the following facts are useful in establishing our results.

\begin{lemma}\label{lem:fieldfacts}
Let $F$ be a field with $4^m$ elements.
\begin{enumerate}[(a)]
\item For nonzero $\alpha\in F$, the polynomial $M(\alpha)M(\alpha^{-1})$ is self-reciprocal.
\item If $\alpha\in F$ is primitive, then $M(\alpha)$ does not divide $X^s-1$ for $s\in [4^m-2]$. 
\item There are $\varphi(4^m-1)$ primitive elements in $F$.
\end{enumerate}
\end{lemma}

Next, we provide a set of polynomials that satisfies Definition~\ref{def:rc-generating}.
\begin{lemma}\label{lem:rc-generating}
Let $g(X)$ be the generator polynomial of a reversible cyclic code $\cB$ of length $n$ and dimension $k$ that contains $1^n$.
Let $\alpha$ be a primitive element of $F$ such that $g(\alpha)\ne 0$ and $g(\alpha^{-1})\ne 0$.
If  $h^*(X)=M(\alpha)M(\alpha^{-1})$ and $n-k<k-1$, 
then the set $\{h^*(X), 1\}$  is $(g,k)$-rc-generating.
\end{lemma}

\begin{proof}We verify conditions (R1) to (R6) in Definition~\ref{def:rc-generating}.

(R1) follows from the fact that both $\alpha$ and $\alpha^{-1}$ are not roots of $g$. 
Since $h^*$ is the product of two minimal polynomials of primitive elements, 
$h^*(1)$ is not zero and so (R2) holds. 
(R3) follows from Lemma~\ref{lem:fieldfacts}(a).

Next, observe that $P=1$ with $p_1(X)=1$. Hence, (R7) trivially holds.
Also, (R4) to (R6) reduces to verifying that 
\begin{enumerate}[(i)]
\item $h^*(X)$ does not divide $X^s-1$ for $s\in [n-1]$; and
\item $h^*(X)$ does not divide $X^{s}-X^{k-1}$ for $0\le s\le n-k$.
\end{enumerate}
Since $n-k<k-1$, we have that $X^{s}-X^{k-1}$ is nonzero for $0\le s\le n-k$.
Then both (i) and (ii) follows from Lemma~\ref{lem:fieldfacts}(b).
\end{proof}

%Finally, we complete the proof.

\begin{proof}[Proof of Corollary~\ref{cor-family-primer}]
Let $g(X)$ be the generator polynomial of the reversible cyclic code $\cC$ constructed in Theorem~\ref{thm-LCDwithone}.

Consider the set $\Lambda = \{ \alpha \in F: g(\alpha)=0 \mbox{ or } g(\alpha^{-1})=0\}$.
Since the degree of $g$ is $n-k\le (d-1)m$, we have that $|\Lambda|\le 2(d-1)m$.
Since $\varphi(4^m-1)> 2(d-1)m$ for $m\ge 6$,
there exists a primitive element $\alpha\in F$ that does not belong to $\Lambda$.
In other words, $g(\alpha)\ne 0$ and $g(\alpha^{-1})\ne 0$.
By Lemma~\ref{lem:rc-generating}, the set $\{h^*(X)\triangleq M(\alpha)M(\alpha^{-1}),1\}$ is $(g,k)$-rc-generating and therefore, Construction~\ref{constr-primer-3} yields an $(n,d;k,k)$-primer code of size $4^{k-2m}$.
\end{proof}

\subsection{Detailed Proofs for Section~\ref{sec:DNAcomputing}}

We complete the proof of Theorem~\ref{thm-DNA} by establishing the following lemma.

\begin{lemma}
Let $g(X)$ be the generator polynomial of a reversible cyclic code $\cB$ of length $n$ and dimension $k$ that contains $1^n$. 
If $\{h(X), p_1(X),  p_2(X), \ldots, p_P(X)\}$ is $(g,k)$-rc2-generating and $k^*=k-\deg h^*$,
then the subcode $\cA=\{(m(X)h^*(X)+1)g(X): \deg m < k^*\}$ satisfies conditions (S1') to (S3') in Lemma~\ref{lem:combinatorial-rc2}. 
\end{lemma}

\begin{proof}Here we only prove condition (S2'). The other two conditions can be proved similarly. 
In particular, we demonstrate that the violation of condition (S2') contradicts condition (R5') in Definition~\ref{def:rc2-generating}.

Suppose to the contrary of (S2') that we have two codewords $\vu,\vv\in\cA$ 
such that $\shift^s(\vu)=\vv^{r}$ for some $s\in\bbracket{n}$. 
Hence, there exists two polynomials $m(X)$ and $m'(X)$ with degrees strictly less than $k^*$, 
two polynomials $p_i(X)$ and $p_j(X)$ with $i,j\in [P]$
such that the following equality holds with some polynomial $f(X)$.
\begin{align*}
&X^s\left(m(X)h^*(X)+p_i(X)\right)g(X)\\
&\hspace{10mm}= X^{n-1}\left(m'(X^{-1})h^*(X^{-1})+p_j(X^{-1})\right)g(X^{-1})\\
&\hspace{65mm} + f(X)(X^n-1).
\end{align*}

As before, by choice of $g$ and $h$, we have that $X^{n-k}g(X^{-1})=g(X)$ and 
$h^*(X)=X^{\deg h^*}h^*(X^{-1})/h^*(0)$.
Dividing the equation by $g(X)$ and rearranging the terms, we have the following equality.
\begin{align*}
&\left(X^sp_i(X) - X^{k-1}p_j(X^{-1})\right)\\
&\hspace{10mm}+ \left(m(X) - X^{k^*}  m'(X^{-1}) h(0)\right)h^*(X) = f(X)h(X)
%&\hspace{40mm}= f(X)h(X).
\end{align*}
%\begin{align*}
%&X^s p_i(X) - X^{k-1}p_j(X^{-1})\\
%&\hspace{2mm}= \frac{h(X)}{X-1}+ f(X)h(X)- \left(X^s m(X) - X^{k^*}  m'(X^{-1}) h(0)\right)h^*(X).
%\end{align*}
%Since $h^*(1)\ne 0$, we have that $h^*(X)$ divides $h(X)/(X-1)$.
Therefore, 
\[h^*(X)\mbox{ divides } p_i(X) - X^{k-1}p_j(X^{-1}),\]
contradicting condition (R5') in Definition~\ref{def:rc2-generating}.
\end{proof}

To complete the proof of Corollary~\ref{cor-family-computing},  we provide a set of polynomials that satisfies Definition~\ref{def:rc2-generating}.

\begin{lemma}\label{lem:rc2-generating}
Let $g(X)$ be the generator polynomial of a reversible cyclic code $\cB$ of length $n$ and dimension $k$ that contains $1^n$.
Let $\alpha$ be a primitive element of $F$ such that $g(\alpha)\ne 0$ and $g(\alpha^{-1})\ne 0$.
If  $h^*(X)=M(\alpha)M(\alpha^{-1})$ and $p(X)=M(\alpha)$, 
then the set $\{h^*(X), p(X)\}$  is $(g,k)$-rc2-generating.
\end{lemma}

\begin{proof}We verify conditions in Definition~\ref{def:rc2-generating}.
Conditions (R1) to (R4) and (R6) follows directly from the proof of Lemma~\ref{lem:rc-generating}.

Hence, we verify (R5') which reduces to verifying that 
\[
h^*(X) \mbox{ does not divide }X^{s}p(X)-X^{k-1}p(X^{-1}) \mbox{ for }s\in\bbracket{n}.
\]
This is equivalent to showing that $r(X)\triangleq X^{s}p(X)-X^{k-1}p(X^{-1})$ is nonzero for some root of $h^*(X)$.
Observe that since $\alpha$ is primitive, we have that $\alpha^{-1}$ is not a root of $p(X)$. 
In other words, $p(\alpha^{-1})\ne 0$.
Since $h^*(\alpha)=p(\alpha)=0$ and $p(\alpha^{-1})\ne 0$, 
we have that $r(\alpha)=\alpha^{k-1} p(\alpha^{-1})\ne 0$. 
\end{proof}

%\begin{table*}
%{\footnotesize
%\begin{tabular}{|p{0.8cm}|p{3.5cm}|p{4.5cm}|p{8.5cm}|}
%\hline
%{Constr.} & Input & Output & Remarks \\ \hline \hline
%
%A &
%$[n, k, d]_2$-cyclic code &
%Balanced $(n+1,2\ceil{d/2})_2$-code \newline
%of size $\ge 2^k/n$ &
%Family of balanced $(n,d)_2$-codes with redundancy $(t+1)\log n+1$,
%where $t=d/2-1$.\newline
%({\em c.f.} Corollary~\ref{cor:bin-balanced}) \\\hline
%
%B &
%$[n+p, n, d]_2$-linear code\newline
%$[n, k, d]_2$-linear code &
%\GC-balanced $(n,d)_4$-code \newline
%of size $\ge 2^{n+k-p}/n$&
%Family of  \GC-balanced $(n,d)_4$-codes with redundancy $(2t+1)\log n+1$,
%where $t=\ceil{(d-1)/2}$.\newline
%({\em c.f.} Corollary~\ref{cor:gc-balanced}) \\\hline
%
%C &
%$[rf+p,rf,d]_2$-linear code \newline
%$\ell$-APD-constrained code of length $f$ &
%An $(n,d;1,2f)$-primer code of length $n=rf+p+\floor{p/(\ell-1)}+\ell+2$ & 
%
%Family of $(n,d;1,f)_4$-primer codes with $n=rf+o(f)$ and
% redundancy at most $(2t+1) \log_4 n +O(1)$, where $t=\ceil{(d-1)/2}$. \newline
%Family of balanced $(n,d;1,f)_4$-primer codes with $n=rf+o(f)$ and
% redundancy at most  $(d+1)\log_4 n +O(1)$.\newline
%({\em c.f.} Corollary~\ref{cor:primer-1}) \\\hline
%\end{tabular}
%}
%\caption{Summary of Code Constructions}
%\end{table*}

%Similar to the proof of Corollary~\ref{cor-family-primer}, we apply Lemma~\ref{lem:rc-generating} to the reversible cyclic code $\cC$ constructed in Theorem~\ref{thm-LCDwithone}. We then obtain the desired family of DNA computing codes for Corollary~\ref{cor-family-DNA}.

\section{Conclusion}

We provide efficient and explicit methods to construct 
balanced codes, primer codes and DNA computing codes with error-correcting capabilities.
Using certain classes of BCH codes as inputs, 
we obtain infinite families of $(n,d)_q$-codes satisfying our constraints
with redundancy $C_d \log n + O(1)$.
Here, $C_d$ is a constant dependent only on $d$ and 
we provide a summary of our constructions and the corresponding value of $C_d$ 
in Table~\ref{table:summary}.
Note that in all our constructions, we have $C_d\le d+1$. 
On the other hand, the sphere-packing bound requires $C_d \geq \floor{(d-1)/2}$.
Therefore, it remains open to provide efficient and explicit constructions that reduce the value of $C_d$ further. 
\bibliographystyle{IEEEtran}
\bibliography{DNA}

%\bibliography{BPC-ref}

\end{document}